\documentclass{article}
\usepackage{fullpage}
\usepackage{amsmath,amssymb,amscd,amsxtra, amsfonts, amsthm, amsbsy,wasysym,graphics, graphicx}
\usepackage{epstopdf}
\usepackage{algpseudocode}
\usepackage{algorithm}
\usepackage{latexsym}
\usepackage{float}
\usepackage{color}

\newcommand{\bitem}{\begin{itemize}}
\newcommand{\eitem}{\end{itemize}}
\newcommand{\benum}{\begin{enumerate}}
\newcommand{\eenum}{\end{enumerate}}
\newcommand{\beq}{\begin{equation}}
\newcommand{\eeq}{\end{equation}}

\newtheorem{theorem}{Theorem}[section]
\newtheorem{pro}[theorem]{Proposition}
\newtheorem{conj}[theorem]{Conjecture}
\newtheorem{lm}[theorem]{Lemma}

\theoremstyle{definition}
\newtheorem{definition}[theorem]{Definition}

\theoremstyle{remark}
\newtheorem{remark}[theorem]{Remark}
\numberwithin{equation}{section}

\newcommand{\Eb}{\mathbb E}
\newcommand{\tr}{\textrm{Tr}}

\begin{document}
\author{Rongrong Wang}
\title{Sigma Delta quantization with Harmonic frames and partial Fourier ensembles}
\maketitle
\begin{abstract}
 Sigma Delta ($\Sigma\Delta$) quantization, a quantization method which first surfaced in the 1960s, has now been used widely in various digital products  such as cameras, cell phones, radars, etc.  The method samples an input signal at a rate higher than the Nyquist rate, thus achieves great robustness to quantization noise. 
 
 Compressed Sensing (CS) is a frugal acquisition method that utilizes the possible sparsity of the signals to reduce the required number of samples for a lossless acquisition.  One can deem the reduced number as an effective dimensionality of the set of sparse signals and accordingly, define an effective oversampling rate as the ratio between the actual sampling rate and the effective dimensionality.  A natural conjecture is that the error of Sigma Delta quantization, previously shown to decay with the vanilla oversampling rate, should now decay with the effective oversampling rate when carried out in the regime of compressed sensing. Confirming this intuition is one of the main goals in this direction. 
 
The study of quantization in CS has so far been limited to proving error convergence results for Gaussian and sub-Gaussian sensing matrices, as the number of bits and/or the number of samples grow to infinity. In this paper, we provide a first result for the more realistic Fourier sensing matrices. The major idea is to randomly permute the Fourier samples before feeding them into the quantizer. We show that the random permutation can effectively increase the low frequency power of the measurements, thus enhance the quality of $\Sigma\Delta$ quantization. 
  
\end{abstract}

\section{Introduction}
Over the last decade, Sigma Delta quantization has been extensively studied in quantizing  band-limited functions  \cite{daub-dev,DGK10, G-exp, inose1963unity} and redundant samples of discrete signals  \cite{benedetto2006sigma,blum:sdf,BPA2007,KSW12,iwen,lammers:adf,PSY13}.   It has also been investigated under the setup of Compressed Sensing \cite{GLPSY13,SWY, KSY13,Feng2014} for (sub)-Gaussian matrices. The popularity of Sigma Delta primarily comes from its strong robustness to quantization errors by oversampling the objective signal and exploiting the correlations among the samples. If the sample correlations are not considered, we arrive at the traditional Memoryless Scalar Quantization (MSQ). MSQ quantizes each sample independently  by directly recording their binary representations. This independent encoding has led to the scheme's instability  under measurement noise as well as quantization errors.  As a result, MSQ are not extensively used in high accuracy analog-to-digital converters (ADC) although it offers the optimal bit efficiency (i.e., bit rate distortion) when no redundancy exists.

Besides the stability, Sigma Delta quantization was also used to overcome hardware limitations. One notable example is its use in certain digital camera designs to enhance image quality captured under low-light. Despite the ever increasing resolution one witnessed,  there were technical difficulties in increasing the dynamic range of the digital camera. The dynamic range is the ratio between the maximum and minimum measurable light intensities.  Compared to the film camera, the dynamic range of the digital cameras was much smaller due to the limitation of its sensing devices, the Complementary metal-oxide-semiconductor (CMOS)  or the Charge-coupled Device (CDD). A direct manifestation is that, when there is not enough light,  the digital camera tends to produce many noisy spots in the photo.  Technically speaking, in this case, the voltage (or current) converted from light intensity via the photodiode is too low compared to quantization and other type of noise. The architecture of Sigma Delta quantization has a feedback loop in it  to perform the oversampling. This feedback loop can be used in this case to accumulate the voltage over several samples hence magnify the signal and increase the signal-to-noise ratio. We refer the interested readers to \cite{Maricic} for an explanation of the exact design of this scheme as well as a demonstration of its efficacy in increasing the dynamic range.

Compressed Sensing is an effective acquisition method for sampling sparsely structured signals. It is followed by a signal inversion procedure including (approximately) solving an $\ell_0$ or $\ell_1$ minimization problem.  Due to the high complexity of the available solvers, this step must be carried out in the computer. Therefore an analysis of how it works with various quantization methods is inevitable. When using Sigma Delta quantization, we naturally ask the question: can the method utilize oversamplings to compress quantization errors as it does for the band-limited functions and if so, what is the relation between the error compression rate and the oversampling rate?  This question has been answered affirmatively in \cite{GLPSY13,SWY, KSY13,Feng2014} for sub-Gaussian matrices with descent error compression rate provided. The proof techniques in these papers apply for all sizes of sensing matrices and bit-depths (the number of  bits allocated for each sample). This paper aims to extend this result to Fourier sensing matrices. 

A parallel line of research investigates MSQ in compressed sensing with only 1 bit of storage per measurement \cite{ALPV14,jacques2013robust, KSW14,PV13, PV13b}. Most results in this direction do not easily generalize to multi-bit scheme, even though having more bits normally reduces the quantization error. In addition, it has been shown that the reconstruction error of 1-bit MSQ is asymptotically larger than that of 1-bit Sigma Delta. This will be discussed in detail later.

To make the illustration  more concrete, we now introduce the mathematical model. Denote the set of all $k$ sparse signal by $\Sigma_k= \{x\in \mathbb{C}^{N}, |\textrm{supp}(x)|=k\}$. Suppose the objective signal $x \in \mathcal{X} \subseteq \mathbb{C}^N $ lies in $\Sigma_k$, and the measurements vector $y$ is obtained under linear map $A_m \in \mathbb{C}^{m,N}$, i.e., $y=A_mx$.  

We say a sensing matrix $A$ satisfies the $\delta_k$-Restricted Isometry Property (RIP) if  
\[
(1-\delta_k) \|x\|_2^2 \leq \|Ax\|_2^2 \leq (1+\delta_k) \|x\|_2^2, \quad \textrm{ for all } x\in \Sigma_k,
\]
and the smallest $\delta_k$ that satisfies the above inequality is called the Restricted Isometry Constant (RIC).

A well-known sufficient condition for the $\ell_1$ minimization problem
\[
\min\limits_z \|z\|_1 \quad \textrm{subject to } Az=y,
\]
to successfully recover all $k$-sparse vectors requires $A$ to satisfy the $(\sqrt2 -1)$-RIP \cite{C08}. Tighter condition was also developed in \cite{CZ14}.   The Fourier ensemble studied in this paper refers to the matrix $A$ being consisted of randomly subsampled rows of the discrete Fourier matrix in $\mathbb{C}^N$. This type of matrices are known \cite{RV08} to satisfy the $\delta_k$-RIP with probability $1-5e^{-\frac{\delta_k^2 t}{C}}$ provided $m\geq ctk \log^4 N $ for any $t>0$, where $c$ and $C$ are both pure constants. Here the lower bound $ctk \log^4 N$ can be considered as the effective dimensionality of the $k$-sparse signals associated with the Fourier ensemble.

Denote by $Q$ the quantization operator that digitalizes $y$ 
  \begin{align*}
  Q: \  \mathbb{C}^m &  \rightarrow \mathcal{A}^m , \\ 
   y&  \rightarrow q,
  \end{align*}
 where $\mathcal{A}$ is called the quantization alphabet, which is the computer's dictionary.  All values in $\mathcal{A}$ must have been assigned digital labels so that they can be accurately recorded in the computer with finitely many bits. 
 
 Now we define the two types of quantizations.
A  scalar quantization $Q_s$ is a simple rounding off operation.
  \begin{align*}
Q_{s}: \mathbb{C} &\rightarrow \mathcal  {A} \\
c & \rightarrow q \quad  \textrm{with } q=  \min\limits_{a\in  \mathcal{A}}  |a-c|.
\end{align*}
The MSQ, $Q_{MSQ}$, of the measurement vector $y$, is defined as 
  \begin{align*}
Q_{MSQ}: A_m(\mathcal{X}) &\rightarrow \mathcal  {A}^m \\
y & \rightarrow q \quad  \textrm{with }  q_i=  Q_s (y_i), \ \ i=1,...,m.
\end{align*}
where $A(\mathcal{X})$ is the space of measurement. Each $q_i$ is acquired independently using only the $i$th sample $y_i$. The method ignores the possible correlation between samples.

The Sigma Delta quantization with order $r$, $Q_{\Sigma\Delta}^r$, is defined as

\begin{align*}
Q_{\Sigma\Delta}^r  : A_m(\mathcal{X})   &\rightarrow \mathcal {A}^m \\
 y & \rightarrow q, \\ 
 \textrm{where} \ \  &  q_i = Q_s\left( \sum\limits_{j=1}^r (-1)^{j-1} \binom{r}{j} u_{i-j}+y_i\right) \notag \\
& u_i =  y_i  -  q_i + \sum^{r}_{j=1} {r \choose j} (-1)^j u_{i-j}   \ \ i=1,...,m \\
& u_0  =c_0, \ \   
\label{equ:rthOrdu}
\end{align*}
where $c_0$ is some pre-assigned constant. A usual choice of $c_0$ is 0. The method records the quantization error in the so-called state variable $u_i$. The next quantization $q_{i+1}$ will leverage this error and the new incoming sample $y_{i+1}$ so that the common information contained in several adjacent samples are refined. For interested readers, the architecture of $\Sigma\Delta$ quantization is explained at length in \cite{GLPSY13,SWY}. 

For simplicity, in this paper we use a uniform alphabet $\mathcal{A} := (\delta \mathbb{Z} +i \delta \mathbb{Z} ) \cap \mathbb{B}(R)$. Here $i$ is the imaginary unit, and $\mathbb{Z}$ is the set of integers. $\delta>0$ defines how dense the alphabet is so is called the quantization step size.  $\mathbb{B}(R)$ denotes the $\ell_2$ ball in $\mathbb{C}$ with center $(0,0)$ and radius $R$.  We assume $R\geq 2 \lceil \|y\|_{\infty} \rceil  +\delta(2^r+1) $ where $\lceil \|y\|_\infty \rceil$ is the smallest integer greater than or equal to $ \|y\|_\infty$. This assumption guarantees the quantization error $\|q-y\|_{\infty}$ to be bounded by $\delta/2$, thus fulfilling the stability requirement of quantization (see \cite{GLPSY13}).

Quantizations are essentially a type of encoders that map real or complex values to a finite codebook of the digital device thus decoders are needed. When $\mathcal{A}$ has large bit-depth, i.e., each sample is recorded by  many bits, then $q$ itself is already a good approximation of $y$ and we can direct use $q$ to recover $x$.  However, for small bit-depth, a decoder needs to be used to find an estimate  $\hat{y}$ of $y$ from $q$ then another decoder is to invert/reconstruct the original signal from $\hat{y}$.  In many cases including ours, the two decoding steps are integrated into a single algorithm, for which the input is $q$ and output is $\hat{x}$, an estimate of $x$.  The performance of these algorithms are evaluated based on their accuracy, efficiency, robustness, etc, which we explain as follows. Efficiency refers to an algorithm's tractability. Robustness refers to the algorithm's stability to additive noise and quantization error. Accuracy is measured by the maximal reconstruction error, i.e., distortion, of all signals among a given class $\mathcal{X}$, 
$$
\mathcal{D}_m: = \max_{x\in \mathcal{X} }  \| x- \Delta(Q(A_m x)\|.
$$  
Here $Q$ denotes a quantization method and $\Delta$ denotes a reconstruction algorithm.
 Previous works \cite{GLPSY13,SWY}  have demonstrated that for Gaussian matrices, the Sigma Delta encoder coupled with several decoders are both robust and efficient. The same is true for our setting by a straightforward generalization. So we only focus on evaluating the accuracy, i.e., finding a formula for the distortion $\mathcal{D}_m$ as a function of $m$. The accuracy for Gaussian matrices has been found in \cite{KSW12, SWY, jacques2013robust}. With the $r$-th order Sigma Delta quantization, we have 
\begin{equation}\label{eq:rate}
\mathcal{D}_m\leq Cm^{-r+1/2},
\end{equation}
 and by optimizing over $r$ it achieves
 $$\mathcal{D}_m \leq e^{-c \sqrt m}.$$  The constants in both expressions depend on $r$ and $k$, the order of $\Sigma\Delta$ and the signal class's sparsity level. Both are assumed to be fixed beforehand. 

In contrast, the distortion for one bit compressed sensing with Gaussian measurement matrix  is known \cite{baraniuk2014exponential,CG15} to at best be 
$$
\mathcal{D}_m\leq Cm^{-1/2}
.$$ 

 There exist more sophisticated quantization schemes that achieve better, exponential asymptotic rate \cite{CG15}, but they come at the expense of complicating the design of the analog hardware of the quantizer, such as involving analog multipliers (devices that perform analog multiplications). In contrast, Sigma Delta quantization only requires analog additions and negations. 

In this paper, we prove similar results as \eqref{eq:rate} for Fourier sensing matrices, which require a complete different technique. As a recursive scheme, it is not very surprising that Sigma Delta displays a great sensitivity to permutations of the input sequence. In fact, we demonstrate that one needs to randomly permute the entires of the Fourier samples in order to obtain uniformly good reconstruction results for all the signals of a given sparsity level. 
\subsection{Notation}
For a matrix $A\in \mathbb{C}^{m,n}$, we use $\|A \|_p$ to denote the Schatten p-norm and when $p=\infty$,  we simply use $\|A\|$ to denote its  spectral norm. In addition, $\textrm{Tr}(A)$ is the trace of $A$, $\|A\|_{\max}$ stands for the entry-wise L-infinity norm, and  $\|A\|_{\ell_2 \rightarrow \ell_{\infty}}$ stands for the operator norm from the $\ell_2$-space to the $\ell_{\infty}$-space. As usual, $\sigma_{\min}(A)$ and $\sigma_{\max}(A)$ are $A$'s smallest and largest singular values, respectively.  If $T$ is an index set $T\subseteq \{1,...,n\}$, then $A_T$ denotes the collection of columns of $A$ indexed in $T$. For a vector $x$, $\|x\|_p$ stands for the $\ell_p$ norm of $x$. The set of all $k$-sparse signals are $\Sigma_k^N:= \{x\in \mathbb{C}^N: |\textrm{supp}(x)|=k\}$. For simplicity, we also use the convention $[N]:=\{1,...,N\}$ for any positive integer $N$.


\section{The proposed method}
 Sigma Delta quantization is conventionally designed to quantize low frequency signals. By definition, the quantizer has the so-called noise shaping effect because it pushes the quantization noise towards high frequencies which are outside the band of interest. 
The Gaussian measurements of sparse signals are not of low frequency. Rather, they have flat spectrums. The results in  \cite{SWY} are interesting as it shows that $\Sigma\Delta$ also perform reasonably well in quantizing this type of inputs. 
\begin{pro}[\cite{SWY}] Suppose $A$ is a sub-Gaussian matrix with zero-mean and unit-variance. Then there exist a convex decoder $\Delta$, and constants $c_1$, $c_2$ and $c_3$, such that with probability over $1-e^{-c_1k\log N/k}$ on the draw of $A$, the reconstruction errors for all $k$-sparse signals obey  
\begin{equation}\label{eq:decay1}
\|x-\Delta \circ Q_{\Sigma\Delta}^r(Ax)\|_2 \leq c_2\left(\frac{m}{k\log (N/k)}\right)^{-r+1/2},
\end{equation}
provided that $m\geq c_3 k\log (N/k)$.
\end{pro}
Despite stated for sub-Gaussian matrix, a close examination of its proof in \cite{SWY}, suggests the above result is obtained by utilizing only the spectrum properties of sub-matrices of $A$, i.e., the bounds on the quantities 
 \begin{equation}\label{eq:quantity}
 \min_{K \subseteq [N], |K|=k}\sigma_{\min}(V_{l}^T A_K), \quad \max_{K \subseteq [N], |K|=k}\sigma_{\max}(V_{l}^T A_K).
 \end{equation} 
 Here $l$ is an integer depending only on $k$, $A_K$ is the collection of columns of $A$ indexed in $K$,  and $V_l$ is the last $l$ left singular vectors of the finite difference matrix. Clearly, since $V_l$ are singular vectors, $V_l^TA_K$ then stands for a projection of $A_K$ onto the span of $V_l$. The following    proposition shows $V_l$ essentially contains low frequency sinusoids. Therefore, the above quantities measure the level of low frequency energy in $A_K$.


\begin{pro}(\cite{KSW12})\label{pro:eigD}
Let $D$ be the $m\times m$ discrete difference matrix with ones and minus ones on the main and sub diagonals, respectively, and zeros otherwise.
Let $USV^T$ be the singular value decomposition of $D^{-1}\in \mathbb{R}^{m\times m}$. Then  \[
V(k,l)=\sqrt{\frac{2}{m+1/2}}\cos\left(\frac{2(l-1/2)(m-k+1/2)\pi}{2m+1}\right),
\]
and
\[
S(k,k)=2\cos\left(\frac{l\pi}{2m+1} \right).\]
\end{pro}
The following propositions are extracted from \cite{SWY}. They show that in the frame case (i.e., $A$ is a tall matrix with full column rank), the distortion is controlled by the amount of  low frequency energy of the analysis frame operator; in the compressed sensing case, the error is determined by the Restricted Isometry Constant of $V_l^TA$. 
 
\begin{pro}\label{pro:2.1} Let $F$  be an $m\times k$ matrix with normalized rows. Then there exists a decoder, such that  for any $x\in \mathbb{B}^k$ (the unit $\ell_2$ ball in $\mathbb{C}^k$), the reconstruction $\hat{x}$ from the quantization $q=Q_{\Sigma\Delta}^r(Fx)$ using this decoder obeys 
\begin{equation}\label{eq:error_bound}\|\hat{x}-x\|_2 \leq c\delta \left(\frac{m}{l}\right)^ {-r} \frac{\sqrt m}{ \sigma_{\min} (V_l ^T F)},
\end{equation}
for any $l$ with $k \leq l \leq m$. Here $c$ is an absolute constant and $\sigma_{\min}(V_l^TF)$ denotes the smallest singular value of $V_l^TF$.  
\end{pro} 

This proposition says that in the frame measurement case, small distortion is guaranteed by large values of $\sigma_{\min}(V_l^TF)$.

\begin{pro}\label{pro:2.2} Let $A$  be an $m\times N$ matrix.  Suppose that there exist some constant $c_0$ and integer $l$ such that $c_0V_l^T A$ has $k$-RIP with RIC $\delta_k <1/9$, i.e.,
\[
\|c_0^2 A_K^T V _l V_l ^T A_K-I_k\|_2 \leq \delta_k,  \textrm{ for all  } K\subset  \{1,..,N\}, |K| \leq k.
\]
Then there exists a decoder, such that  for any $k$-sparse signal $x$, the reconstruction $\hat{x}$ from the quantization $q=Q_{\Sigma\Delta}^r(Ax)$ using this decoder obeys 
\begin{equation}\label{eq:error_bound1} \|\hat{x}-x\|_2 \leq c\delta \left(\frac{m}{l}\right)^ {-r} \sqrt m. \end{equation}
Here $c$ is a constant that only depends only on $c_0$.  
\end{pro} 
This proposition proves polynomial convergence rate for the distortion under the condition that  the quantity $V_l^TA$ satisfies a RIP condition. The condition implies that all $A_K$ with $|K|=k$ have similar level of low frequency energy. 




Consider the situation when $A$ is the DFT matrix, then there exist choices of $K$ that $A_K$ only contains high frequency sinusoids, and also cases where $A_K$ contains only low frequency sinusoids. Together they induce a large RIC bound on $V_l^TA$, in fact too large for Proposition 2.4 to be informative. Even when the support is known in which case the measurement reduces to $y=A_K x_K$, Proposition 2.3 would yield sub optimal result due to the small value of $\sigma_{\min}(V_l^TA_K)$ in the case when $A_K$ only contains high frequency components.

To overcome this intrinsic issue of Fourier ensembles, and inspired by the success of the sub-Gaussian ensembles, we propose to randomly permute the entries of the Fourier measurements so as to whiten its spectrum before letting it enter the quantizer. The random permutation we consider in this paper is with replacement. The result for random permutation  without replacement is similar but with more complicated proof  so we omit it to avoid distraction from the main topic. 

  
In what follows, we first demonstrate the random permutation increases the value of $\sigma_{\min}(V_l^TF_K)$ so we have the desired distortion bound for the frame case. The result for the compressed sensing case then follows from it. 

We start with a definition of the harmonic frame as  sub-matrices of the DFT matrices.
\begin{definition}[Harmonic Frames]
An $m\times k$ harmonic frame $H^{\omega}$ has the form 
\begin{equation}\label{eq:harmonic}
H^{\omega} (j,l)= e^{2\pi i \omega(l) \frac{j}{m}}, \ \  \text{ for } j=1,..., m, \ l=1,..,k,  
\end{equation}
where $\omega \subseteq \{1,..., m\}.$
\end{definition}

By definition,  all measurements under $H^\omega$ are band-limited signals. Denote by $\sigma : \mathbb{Z}^m \rightarrow \mathbb{Z}^m$ the map that creates a new measurement sequence via randomly selecting elements from $y$ with replacement. $\sigma(i)$ stands for the index chosen in the $i$th random draw, i.e., $\sigma(1)=3$ means that the first element in the new sequence is the third one in the original. We denote the new sequence by $y_{\sigma}:= H^{\omega}_\sigma x$ hence $H_\sigma$ is the matrix formed by randomly selecting rows of $H$ according to $\sigma$. The next theorem shows that with large probability on the draw of $\sigma$,   all measurements under $H^{\omega}_{\sigma}$ have a somewhat uniform spectrum. 

 \begin{theorem}\label{thm:freq} Fix integers $N \geq l\geq k$, an absolute constant $\epsilon<1$, and two index sets $\omega_1, \omega_2, \subseteq \{1,..., N\} $  with $|\omega_1|=l$, $|\omega_2|=k$.  Let $\sigma$ be the random selection map $\sigma:\mathbb{Z}^N \rightarrow \mathbb{Z}^N$ defined by selecting indices from the set $\{1,...,N\}$ randomly $N$ times, i.e., for each $i,j=1,..,N$, $P(\sigma(i)=j)=\frac{1}{N}$.  Then there exist constants $c_1$ and $c_2$, such that,  with probability $\epsilon$ on the draws of $\sigma$, it holds, $$\sigma_{\min} ((\bar H^{\omega_1})^T \bar H^{\omega_2}_{\sigma} ) \geq  \min  \left\{ c_1 \sqrt{\frac{l}{N}}, 1-(1- \delta_1(\omega_1) )\delta_1 (\omega_2)  \right\} ,$$ and
$$\sigma_{\max} ((\bar H^{\omega_1})^T \bar H^{\omega_2}_{\sigma} )\leq \max\{c_2 \sqrt{\frac{l}{N}},\delta_{1}(\omega_1) \delta_{1}(\omega_2)\},$$
provide that $ l\geq c k \log^3 (N/ \epsilon)$. Here $\bar H^{\omega}$ denotes $H^{\omega}$ with normalized columns and $H^{\omega}$ denotes a harmonic frame in $\mathbb{C}^N$. $\bar H^{\omega}_{\sigma}$  is the matrix formed by stacking rows of $\bar H^{\omega}$ according to the order defined in $\sigma$. $\delta_1(\omega)$ is an indicator function that indicates whether the first column (the all-one column) belongs to $H^\omega$.
\end{theorem}

If all the frequency bands carry the same level of energy, then the relative energy in arbitrary $l$ out of $N$ bands should be $\sqrt{l/N}$. Hence the upper and lower bound in Theorem 1.5 are tight up to a constant when the column of ones is not in $H^{\omega_2}$ .  

The following theorem replaces $H^{\omega_1}$ in the above theorem by $V_l$, providing a direct estimate of the key quantity  $\sigma_{\min}(V_l^T H^{\omega})$. 

 \begin{theorem}\label{thm:main2}  Let $H$ be any $N\times k$ harmonic frame defined in \eqref{eq:harmonic}. Let $F$ be an $m\times k$ matrix whose rows are randomly chosen from those of $H$ with replacement. Then there exist $c_1>0$ such that for any $\eta, \epsilon>0$,  any $l$  and $m$ satisfying $m/\pi^2 \geq l\geq \frac{c_1}{\eta^2} k\log^3 (m/ \epsilon)$, it holds
\begin{equation}\label{eq:RIP1}
P\left(\sigma_{\min} (V_l^T F)\leq (1-\eta) \sqrt l \right) < \epsilon,
\end{equation}
where $V_l$ are the first $l$ right singular vectors of $D^{-1}$. 
If in addition, $\Phi$ does not contain the column of all ones, then it also holds that
\begin{equation}\label{eq:RIP2}
P\left(\sigma_{\max}  (V_l^T F)\geq (1+\eta) \sqrt l\right) < \epsilon.
\end{equation}
\end{theorem}

Theorem 2.6 and 2.7 will be proved together as corollaries of the following more general result.  It shows that the only property of $V_l$ needed to obtain results like Theorem 2.6 and 2.7 is the wide-spreadness of its entries. 

  \begin{theorem}\label{thm:main1} Let $F \in C^{N,k}$ be a tight frame with frame bound $N$ (i.e., $F^TF=NI_k$) and assume $e^TF=0$ where $I_k$ is the $k\times k$ identity matrix and $e=[1,...,1]^T\in \mathbb{R}^N$. Let $V=[v_1,...,v_l]$ be some  $m\times l$ orthonormal  matrix. Assume constants $r_1,r_2>0$ are such that $\|F\|_{\ell_2\rightarrow \ell_{\infty}} \leq r_1\sqrt k$ and that $\|V\|_{\ell_2\rightarrow \ell_{\infty}} \leq r_2 \sqrt{\frac{l}{m}}$. 
  Let $F$ be an $m\times k$ matrix whose rows are randomly chosen from those of $\Phi$ with replacement. Then there exists a positive function $c_1$ such that for any $c,\epsilon>0$, as long as $l$ satisfies $m\geq l\geq c_1(r_1,r_2,c) k \log^3 (m/ \epsilon)$, it holds with probability $1-\epsilon$ that
\[
\sigma_{\min}^2(V_l^TF )  \geq l\left(1-c-2r_1^2\sqrt{\frac{k}{l}} \log\frac{4 k}{ \epsilon}\right),
\]
and 
\[
\sigma_{\max}^2(V_l^TF) \leq l\left(1+c+2r_1^2\sqrt{\frac{k}{l}} \log\frac{4 k}{ \epsilon}\right).
\]
where $c_1(r_1,r_2,c)=c_2r_1^4r_2^2/c^2$ with $c_2$ being an absolute constant. 
\end{theorem}

\section{Applications to frame and compressed sensing settings} 
Applying Theorem \ref{thm:main2}, we can show that with some appropriate decoder, the distortion under randomly permuted partial Fourier measurements $q=Q_{\Sigma\Delta} ^1 (A_{\sigma} x) $ obeys
\begin{equation}\label{eq:result}
\|\hat{x}-x\|_2 \leq  C(k,N) m^{-1/2},
\end{equation}
where $C(k,N)$ depends on the signal's sparsity level $k$ and the ambient dimension $N$. 
This error bound has the same asymptotic order in $m$ as that for the sub-Gaussian matrices.
\subsection{Decoders}
The existing decoders in the literature are sufficient to show our results. When $A$ is underdetermined, we use the following decoder proposed in \cite{SWY} to reconstruct $x$ from the $r$th order Sigma Delta quantized measurements $q=Q^r_{\Sigma\Delta} (y_{\sigma})$, 
\begin{equation}\label{eq:decoder}
\hat{x}=\arg\min\limits_{z} \|z\|_1 \textrm{ subject to } \| D^{-r} (A_{\sigma}z-q)\|_\infty \leq \delta/2.\tag{$D_{CS}$}
\end{equation}
When $A$ is overdetermined with full column rank, we can either use  a simplified version of \eqref{eq:decoder}, 
 \begin{equation}\label{eq:consistent} \hat{x}=\arg\min\limits_z 1, \text{ s.t. } \|D^{-r}(A_{\sigma}z-q)\|_\infty \leq  \delta/2, \tag{$D_c$}
 \end{equation}
or the $r$th order Sobolev dual frame proposed in \cite{blum:sdf}, 
 \begin{equation}\label{eq:sob}
 A_{sob,r}= (A^T(D^T)^{-r} D^{-r} A)^{-1} A^T(D^T)^{-r}D^{-r}. \tag{$D_{sob}$}
 \end{equation}
 \subsection{High order Sigma Delta}
To generalize \eqref{eq:result} to higher order $\Sigma\Delta$ quantizations $(r\geq 2$), we need some unjustified properties of the singular vectors of $D^{-r}$ to derive the necessary estimate on $r_2$ in Theorem \ref{thm:main1}.  Explicitly,  from numerical experiments, we have the following conjecture. 
\begin{conj}\label{conj}
There exists a constant $c$ such that for any $r$, the singular vectors $V$ of $D^{-r} \in \mathbb{R}^{m\times m}$ satisfies $$ \|V\|_{\max}\leq cr^r \sqrt{\frac{1}{m}},$$
where $\|V\|_{\max}:=\max\limits_{i,j} |v_{i.j}|$ is the element-wise norm of $V$.
\end{conj}
If this conjecture is true, then the result \eqref{eq:result} can be generalized to $r\geq 2$ as
\[
\|\hat{x}-x\|_2 \leq  C(k,N,r) m^{-r+1/2}.
\]

\subsection{Main results}
We are now ready to state the main theorems of the paper.

%
\begin{theorem}\label{cor:optimal}
Denote by $ \mathbb{B}^k$ the unit $\ell_2$ ball in $\mathbb{C}^k$.  Let $F$  be an $N\times k$ Harmonic frame, and $\widetilde{F} \in \mathbb{C}^{m,k}$ be a matrix with randomly selected rows from $F$ with replacement.  Suppose $x\in \mathbb{B}^k$ is the signal and $q=Q^1(\widetilde{F}x )$ is the first order $\Sigma\Delta$ quantization of $\widetilde{F} x $  using the uniform quantization alphabet $\delta\mathbb{Z}+ i\delta\mathbb{Z}$. Then there exist absolute constants $c_1$ and $c_2$ such that for any $\epsilon >0 $, the reconstruction $\hat{x}$  from $\hat{x}=\widetilde{F}_{sob}(q)$ or from
  \begin{equation*}
 \min_{z} 1, s.t.  \|D^{-1}(q-\widetilde{F}z)\|_\infty \leq \delta/2,
 \end{equation*}
obeys
$$\sup\limits_{x\in \mathbb{B}^k} \|\hat{x}-x\|_2 \leq c_1 \delta \left(\frac{m}{l}\right)^ {-1/2}, $$
with probability $\epsilon$ provided that $m/\pi^2 \geq l\geq c_2 k\log^3 (m/ \epsilon)$.\end{theorem}
\begin{theorem}\label{thm:cs}Let $A$ be an unnormalized DFT matrix of dimension $N$, and let $\widetilde{A} \in \mathbb{C}^{m,N}$ be a matrix with randomly selected rows from $A$ with replacement. 
Assume $x\in \mathbb{C}^N$ is a $k$-sparse signal. Let $q=Q^1(\tilde{A} x)$ be the first order Sigma Delta quantization of the compressed measurements $\tilde{A}x$ with the quantization alphabet $\delta\mathbb{Z}+\delta\mathbb{Z}i$ and suppose $\hat{x}$ is the solution to 
\begin{equation}\label{eq:l1}
\min\|z\|_1, \text{ s.t. } \|D^{-1}(q-\tilde{A} z)\|_{\infty}\leq \delta/2.
\end{equation}
Then there exist absolute constants $c_1$ and $c_2$ such that for any $\epsilon>0$, 
\begin{equation}\label{eq:csbound}
\sup\limits_{x\in \Sigma_k^N} \|x-\hat{x}\|_2 \leq c_1\delta\left(\frac{m}{k^4\log^3 N/\epsilon}\right)^{-1/2},
\end{equation}
with probability over $1-\epsilon$ provided that $m\geq c_2 k^4\log^3 \frac{N}{\epsilon}$. Here  $\Sigma_k^N$ denotes the set of $k$-sparse signals in $\mathbb{C}^N$. 
\end{theorem}

\section{Proofs of the theorems}

\subsection{Auxiliary Lemmas } 
In this section, we list a few large deviation results from probability theory as well as some preliminary lemmas that will be employed later .
\begin{pro}[Bernstein inequality] \label{lm:Bern} Let $X_1,..., X_n$ be independent zero-mean random variables. Suppose that $|X_i| \leq M$ almost surely, for all $i$. Then for all positive $t$,
\[
P\left(\sum\limits_{i=1}^n X_i>t\right) \leq \exp\left(-\frac{\frac{1}{2}t^2}{\sum \Eb [X_j^2]+\frac{1}{3}Mt}\right).
\]
\end{pro}
\begin{pro}[Decoupling \cite{Ver}]\label{lm:decoupling}
 Let $A$ be an $n\times n$ matrix with zero diagonal. Let $X=(x_1,...,x_n)$ be a random vector with independent mean zero coefficients. Then, for every convex function $F$, one has 
\[
\Eb F(\langle AX, X \rangle)=\Eb F(\langle 4 A X, X'\rangle),
\]
where $X'$ is an independent copy of $X$.
\end{pro}
\begin{pro}[Matrix Rosenthal inequality \cite{MT14}.]\label{lm: MR} Suppose that $p=1$ or $p\geq 1.5$. Consider a finite sequence $(Y_k)_{k\geq 1}$ of centered, independent, random Hermitian matrices, and assume that $\Eb \|Y_k\|_{4p}^{4p}<\infty$. Then
\[
\left(\Eb \left\|\sum\limits_k Y_k\right\|_{4p}^{4p}\right)^{1/(4p)} \leq \sqrt{4p-1}\cdot \left\|\left(\sum\limits_k \Eb Y_k^2\right)^{1/2}\right\|_{4p}+(4p-1)\cdot \left(\sum\limits_k \Eb\|Y_k\|_{4p}^{4p}\right)^{1/(4p)}.
\]
\end{pro}
\begin{pro}[Matrix Bernstein \cite{MT14}]\label{lm:MB} Consider an independent sequence $(Y_k)_{k\geq 1}$ of random matrices in $\mathbb{C}^{d_1\times d_2}$ that satisfy $\Eb Y_k=0$ and $\|Y_k\|\leq R$ for any $k$ almost surely. Then, for all $t\geq 0$
\[
P\left(\lambda_{\max}\left(\sum\limits_kY_k\right)\geq t\right)\leq (d_1+d_2) \exp\left(\frac{-t^2}{3\sigma^2+2Rt}\right),
\]
for $\sigma^2=\max\left\{ \left\|\sum\limits_k \Eb Y_kY_k^T \right\|,\left\|\sum\limits_k \Eb Y_k^T Y_k \right\|\right\}.$
\end{pro}
\begin{pro}[\cite{Foucart13}]\label{pro:foucart} Let $f,g\in \mathbb{C}^N$, and $\Phi \in \mathbb{C}^{m,N}$. Suppose that $\Phi$ has $\delta_{2k}$-RIP with $\delta_{2k}<1/9$, then for any $1\leq p\leq 2$, we have
\[
\|f-g\|_p\leq C_1k^{1/p-1/2} \|\Phi(f-g)\|_2+\frac{C_2}{k^{1-1/p}}(\|f\|_1-\|g\|_1+2\sigma_k(g)_1),
\]
with constants $C_1$, $C_2$ only depending on $\delta_{2k}$.
\end{pro}

\begin{lm}\label{lm:cond} Let $B$ be an $m\times m$ diagonal positive definite matrix, and let $E\in \mathbb{C}^{N,k}$ be such that $E^TE=NI_k$, where $I_k$ is the $k\times k$ identity matrix. In addition, suppose $\|E\|_{\ell_2\rightarrow l_{\infty}} \leq r_1\sqrt k$ with some constant $r_1$. Then, for any $m$ such that $m\geq k $, the random $m\times k$ matrix $F$ whose rows are independently and uniformly chosen  from the rows of $E$ satisfies
\begin{equation}\label{eq:f}
P\left( \| F^TBF-\tr(B)I_k\| \geq t \right) \leq 2k \exp\left(-\frac{t^2}{3(r_1^2k-1)\textrm{Tr}(B)+2r_1^2k\|B\|_{\max} t}\right).
\end{equation}
\end{lm}
\begin{proof}
Apply Proposition \ref{lm:MB} to $F^TB F- \tr(B) I_k=\sum\limits_k (b_k f_k f_k^T-b_k I_k)$, where $f_k$ is the $k$-th row of $F$, and $b_k$ is the $k$th element of $B$ on the diagonal.
Observe that $\Eb (f_kf_k^T-I)=0$, $$R=\max\limits_k \|b_kf_kf_k^T-b_kI_k\|\leq \max\limits_k b_k\|f_kf_k^T\|\leq  r_1^2 k \max\limits_k b_k=r_1^2 k \| B\|_{\max}$$ and 
$$\sigma^2=\|\sum\limits_{k=1}^m b_k\Eb (f_kf_k^T-I)(f_kf_k^T-I)\|\leq (r_1^2k-1)\tr(B).$$ Hence Proposition \ref{lm:MB} implies that
\begin{equation}\label{eq:f}
P\left( \|\sum\limits_{k=1}^m b_kf_kf_k^T-\tr(B)I_k\|_2\geq t \right) \leq 2k \exp\left(-\frac{t^2}{3(r_1^2k-1)\tr(B)+2r_1^2k\|B\|_{\max} t}\right)
\end{equation}
\end{proof}
\subsection{Proof of Proposition \ref{pro:2.1} and \ref{pro:2.2}}
As mentioned above, the proof is essentially contained in the proof of Theorem 9 of \cite{SWY}. We extract the key steps and present them here for completeness.
\begin{proof} [Proof of Proposition \ref{pro:2.1}]
We can use either the Sobolev dual decoder or the $\ell_1$ minimization decoder to produce the stated reconstruction error.
(a) If $\hat{x}= F_{sob} q$, then
\begin{align*}
\|\hat{x}-x\|_2 &= \| F_{sob} (q- y)\|_2 = \|F_{sob}\|_2 \|q-y\|_2 \leq \sqrt m \delta \|F_{sob}\|_2 \leq \frac{\sqrt m \delta}{\sigma_{\min}(D^{-r}F)} \\ & \leq \frac{\sqrt m \delta}{\sigma_l(D^{-r})\sigma_{min}(V_l^T F)}  \leq c \delta \left(\frac{m}{l}\right)^{-r}\frac{m}{ \sigma_{\min}(V_l^TF)}, 
\end{align*}
where $\sigma_l(D^{-r})$ denotes the $l$th largest singular values of $D^{-r}$ and we have used the fact that 
\begin{equation}\label{eq:eigD}
c_1\left(\frac{m}{l}\right)^r \leq \sigma_l(D^{-r}) \leq c_2\left(\frac{m}{l}\right)^r,
\end{equation} derived in \cite{GLPSY13}. \\
b) If $\hat{x}$ is a feasible solution to \eqref{eq:consistent}, then by the triangle inequality,
\[
2\sqrt m \delta\geq \|D^{-r}F(x-\hat{x})\|_2\geq c\left(\frac{m}{l} \right)^r\sigma_{\min}(V_l^T F) \|x-\hat{x}\|_2 .
\]
Rearranging the above equation yields the conclusion of the theorem.
\end{proof}
\begin{proof}[Proof of Proposition \ref{pro:2.2}]
Suppose the decoder \eqref{eq:decoder} is used for reconstruction. Due to the feasibility,
  $$\| D^{-r} A( \hat{x}-x)\|_2 \leq \sqrt m \delta.$$ 
Let $USV^T$ be the SVD of $D^{-r}$ and $V_l$ be the first $l$ columns of $V$. Then 
\begin{equation}\label{eq:33}
 \sqrt m \delta \geq \|S_l V_l^T A (\hat{x}-x)\|_2\geq \frac{1}{c_0}\sigma_l(D^{-r} ) \|c_0 V^T_l A(\hat{x}-x)\|_2. 
\end{equation}
The result follows from inserting the lower bound in \eqref{eq:eigD} into \eqref{eq:33} and invoking Proposition \ref{pro:foucart}.
\end{proof}
\subsection{Proof of Theorem \ref{thm:freq}, \ref{thm:main2} and \ref{thm:main1}}
\begin{proof}[Proof of Theorem \ref{thm:main1}]
  
For simplicity, we assume $m$ is even. The odd case follows a similar line of argument. 

Let  $\widetilde{A}$ be the diagonal matrix containing the  main diagonal of the matrix $V_lV_l^T$. By definition $\widetilde{A}$ is positive definite. 

Let $A=V_lV_l^T-\widetilde{A}$,  and $Y=F^TV_lV_l^TF=F^TAF+F^T\widetilde{A}F$. By the normality of $V_l$, we have $\tr(\widetilde{A})=\|V_l\|_F^2= l$, and 
$$\|\tilde{A}\|_{\max} = \|V_l\|^2_{\ell_2 \rightarrow \ell_{\infty}}\leq 1.$$
Here $\| R \|_{\ell_2 \rightarrow \ell_{\infty}}$ equals to the maximum row norm of $R$.
Applying Lemma \ref{lm:cond} with $t= 2r_1^2\sqrt{kl} \log\frac{4 k}{ \epsilon}$, we are led to 
\begin{equation}\label{eq:failure1}
P\left( \|F^T\widetilde{A} F-lI_k\|_2\geq 2r_1^2\sqrt{kl} \log\frac{4 k}{ \epsilon} \right) \leq \frac{\epsilon}{2}.
\end{equation}
The following is devoted to finding an upper bound for the quantity $\|F^TAF\|_2$.
Let $\eta \in \{0,1\}^m$ be a random vector which randomly selects $m/2$ indices from a total of $m$ indices. By a similar argument as in \cite{tropp},
\[
\Eb (\eta_j(1-\eta_i)+\eta_i(1-\eta_j))=\frac{m}{2m-2}.
\]
Using the convexity of the $\ell_p$ norm, we have
\begin{align}\label{eq:decouple}
\Eb \|F^TAF\|_{4p}^{4p}&=\Eb \left\|\sum\limits_{i\neq j} a_{i,j} f_i f_j^T \right\|_{4p}^{4p} \notag \\
&=\Eb \left\|\sum\limits_{i\neq j} \frac{2m-2}{m} \Eb_{\eta}(\eta_j(1-\eta_i)+\eta_i(1-\eta_j)) a_{i,j} f_i f_j^T\right\|_{4p}^{4p} \notag \\
&\leq 2^{4p} \Eb_{\eta} \Eb \left\|\sum\limits_{i\neq j} (\eta_j(1-\eta_i)+\eta_i(1-\eta_j)) a_{i,j} f_i f_j^T\right\|_{4p}^{4p} \notag \\
&=2^{4p} \Eb_{\eta} \Eb \left\|\sum\limits_{i\in T_1,j\in T_2} a_{i,j} f_i f_j^T+\sum\limits_{i\in T_2,j\in T_1}a_{i,j} f_i f_j^T \right\|_{4p}^{4p} \notag \\
&\leq 2^{8p} \Eb_{\eta} \Eb \left\|\sum\limits_{i\in T_1,j\in T_2} a_{i,j} f_i f_j^T \right\|_{4p}^{4p} \notag\\
&=2^{8p} \Eb_{\eta} \Eb \left\|F_{T_1} A_{T1\times T2}F_{T_2}^T \right\|_{4p}^{4p},
\end{align}
where $T_1=\{j, \eta_j=0\}$, and $T_2=\{j, \eta_j=1\}$, $F_{T_1}$ means restricting $F$ to the columns indexed by $T_1$, and $A_{T_1,T_2}$ means restricting to $A$ to the submatrix with indices $\{(i,j),i\in T_1, j\in T_2\}$ .\\

Next we shall calculate $M:=\Eb \left\|F_{T_1}A_{T1\times T2}F_{T_2}^T \right\|_{4p}^{4p}$ for a fixed $\eta$. Since $T_1\cap T_2=\emptyset$, then $F_{T_1}$ and $F_{T_2}$ are now independent of each other.

Set $E=[e_1,..., e_{m/2}]\equiv F_{T_1}A_{T_1\times T_2}$, where $e_i=F_{T_1}a_i$ is the $i$th column of $E$ and $a_i$ is the $i$th column of $A_{T_1\times T_2}$. Moreover, let $f_i$, $i=1,...,m/2$ be columns of $F_{T_2}$, and $Y_i=e_if_i^T$. Then
\begin{align}
M&=\Eb_1\Eb_2 \left\| \sum\limits_{i=1}^{m/2}e_if_i^T \right\|_{4p}^{4p}=\Eb_1\Eb_2 \left\| \sum\limits_{i=1}^{m/2}Y_i \right\|_{4p}^{4p} \notag \\
&\leq 2^{4p}(4p-1)^{2p} \Eb_1 \max \left\{ \underbrace{\left\| \left(\sum\limits_k \Eb_2 Y_kY_k^T\right)^{1/2} \right\|_{4p}^{4p}}_{(I)}, \underbrace{\left\| \left(\sum\limits_k \Eb_2 Y_k^TY_k\right)^{1/2} \right\|_{4p}^{4p} }_{(II)}\right\} \label{eq:rt}\\
&+ 2^{4p}(4p-1)^{4p} \underbrace{\Eb_1 \left( \sum\limits_k \Eb_2\|Y_k\|_{4p}^{4p} \right)}_{(III)}, \notag
\end{align}
where $\Eb_1$ and $\Eb_2$ stands for expectations with respect to $F_{T_1}$ and $F_{T_2}$ respectively. \eqref{eq:rt} used Proposition \ref{lm: MR} applied to the centered,  i.i.d. matrices
$$
\tilde{Y}_k =\left [\begin{matrix} 0 & Y_k \\ Y_k' & 0 \end{matrix} \right], k=1,...,m/2.
$$
Now we compute the bounds for $(I),(II),(III)$ separately.
\begin{align}\label{eq:III}
(III)&=\Eb_1 \left( \sum\limits_k \Eb_2\|Y_k\|_{4p}^{4p} \right)\leq \Eb_1 \sum\limits_k \Eb_2(\|e_k\|\|f_k\|)^{4p} = r_1^{4p}k^{2p} \Eb_1 \sum\limits_k (\|e_k\|)^{4p} \notag \\
&=r_1^{4p}k^{2p}  \Eb_1 \sum\limits_k \|F_{T_1}a_k\|^{4p} \leq r_1^{4p}k^{2p} \Eb_1 \sum\limits_k \|F_{T_1}\|^{4p} \|a_k\|^{4p}.
\end{align}
In the above we have sequentially used the definition $Y_k$, the bound on $\|f_k\|$, the definition of $e_k$ and that of the operator norm. 

On the other hand,
\begin{align}\label{eq:int}
\Eb \|F_{T_1}\|^{4p} &=\Eb \|F_{T_1}^TF_{T_1}\|^{2p} = \Eb \|F_{T_1}^TF_{T_1}-\frac{m}{2}I+\frac{m}{2}I \|^{2p} \notag \\
& \leq 2^{2p-1} \left( \Eb \|F_{T_1}^TF_{T_1}-\frac{m}{2}I\|^{2p}+\left(\frac{m}{2}\right)^{2p}\right) \notag \\
& \leq 2^{2p-1} p \int\limits_{0}^{\infty} t^{2p-1} k e^{-\frac{t^2}{2mr_1^2 k}} dt + 2^{2p-1} p \int\limits_{0}^{\infty} t^{2p-1} k e^{-\frac{t}{2r_1^2 k}} dt+ m^{2p} \\
& \overset{s=t^2}{=}2^{2p-2} pk \int\limits_{0}^{\infty} s^{p-1}  e^{-\frac{s}{2mr_1^2 k}} ds+2^{2p-1} pk \int\limits_{0}^{\infty} t^{2p-1} e^{-\frac{t}{2r_1^2 k}} dt + m^{2p} \notag \\
& \leq 2^{4p-2} p! pk (mr_1^2k)^p+ 2^{2p-1} (2p)! pk(kr_1^2 )^{2p}+m^{2p} \label{eq:gamma}\\
& \leq 2^{4p-2} p^{p+1} k^{p+1} m^p r_1^{2p}+ 2^{4p-1} p^{2p+1} k^{2p+1}r_1^{4p}+m^{2p},  \notag
\end{align}
where \eqref{eq:int} used the fact that $\Eb a^p=p\int\limits_0^{\infty} t^{p-1}P(a>t)dt$ for positive $a$. \eqref{eq:gamma} used the property of gamma function.

Substituting the above bound into \eqref{eq:III}, and  noting that
\[
\|a_k\|_2\leq \|V_l v_k\|_2 \leq \|v_k\|_2\leq r_2\sqrt{\frac{l}{m}},
\] we obtain
\begin{align*}
(III) & \leq   r_1^{4p}k^{2p} m (2^{4p-2} p^{p+1} k^{p+1} m^p r_1^{2p}+ 2^{4p-1} p^{2p+1} k^{2p+1}r_1^{4p}+m^{2p}) r_2 ^{4p}\left(\frac{l}{m}\right)^{2p}\\ &=c_1^pr_2^{4p} \left(\frac{p^{p+1} k^{3p+1}l^{2p}r_1^{6p}}{m^{p-1}}+\frac{p^{2p+1} k^{4p+1}l^{2p}r_1^{8p}}{m^{2p-1}}+m k^{2p}l^{2p}\right).
\end{align*}
To bound $(I)$, notice that
\begin{equation}\label{eq:I}
(I)=\Eb_1\left\|\left(\sum\limits_k \Eb_2 (Y_kY_k^T)\right)^{1/2} \right\|_{4p}^{4p}= \Eb_1 \left\|\left(\sum\limits_k\Eb_2 e_kf_k^T f_k e_k^T \right)^{1/2}\right\|_{4p}^{4p}\leq r_1^{4p} k^{2p} \Eb_1 \|E\|_{4p}^{4p}\leq r_1^{4p} k^{2p+1}  \Eb_1 \|E\|_2^{4p},
\end{equation}
where we have used the assumption $\|F\|_{\ell_2\rightarrow \ell_{\infty}} \leq r_1\sqrt k$ , and the definition of $E$, $E=\sum\limits_{i=1}^{m/2} f_ia_i^T$. Applying Proposition \ref{lm:MB} to $E$, we get 
\begin{equation}\label{eq:dv}
P(\|E\|\geq t) \leq m \exp\left(-\frac{t^2}{3\sigma+2Rt}\right).
\end{equation}
Here $R=\max\limits_{i=1,...,m/2} \|f_ia_i^T\|_2 \leq r_1r_2\sqrt{\frac{kl}{m}}$ and 
\begin{align*}
\sigma^2 &=\max\left\{ \left\|\sum\limits_{i=1}^{m/2} \mathbb{E} f_i a_i ^T a_i f_i^T\right\|_2,  \left\|\sum\limits_{i=1}^{m/2} \mathbb{E} a_i f_i ^T f_i a_i^T\right\|_2\right\}  \\
 & = \max\left\{  \frac{1}{2}\left\|\sum\limits_{i=1}^{m/2} a_i^T a_i \right\|, r_1^2 k\left\|\sum\limits_{i=1}^{m/2} a_i a_i^T\right\|_2\right \} \\
 & = \max\left\{ \frac{1}{2}\|A_{T_1\times T_2}\|_F^2, r_1^2k\|AA^T\|_2 \right\} \\
 & =\max\{\frac{l}{2},r_1^2 k\}= \frac{l}{2}.
\end{align*}
The last equality is due to the assumption  $l\geq c_1(r_1,r_2,c) k \log^3 (m/ \epsilon)$ and the fact $r_1 \geq 1$ that is a direct consequence of its definition.
Inserting the value of $R$ and $\sigma^2$ into \eqref{eq:dv}, we get
\[
P(\|E\|\geq t) \leq m \exp\left(-\frac{t^2}{\frac{3l}{2}+2r_1r_2\sqrt{\frac{kl}{m}}t}\right) \leq m \exp\left(-\frac{t^2}{3l}\right)+m \exp\left(-\frac{t}{4r_1r_2\sqrt{\frac{kl}{m}}}\right).
\]
This implies 
\begin{align}
\mathbb{E} \|E\|_2^{4p} &\leq  4p \int\limits_0^{\infty} t^{4p-1}me^{-\frac{t^2}{3l}}dt+4p \int\limits_0^{\infty} t^{4p-1}me^{-\frac{t}{4r_1r_2\sqrt{kl/m}}}dt  \notag \\
&= c_2^pm\left(p^{2p} l^{2p}+p^{4p}r_1^{4p}r_2^{4p}\left( \frac{kl}{m}\right)^{2p}\right).
\end{align}
Hence 
\[
(I) \leq c_2^pmr_1^{4p}k^{2p+1}\left(p^{2p} l^{2p}+p^{4p}r_1^{4p}r_2^{4p}\left( \frac{kl}{m}\right)^{2p}\right).
\]
Last observe that $(II)$ is bounded by $(I)$ as
\begin{align}\label{eq:II}
(II)&=\Eb_1 \left\|\left(\sum \Eb_2(f_ie_i^Te_if_i^T) \right)^{1/2}\right\|_{4p}^{4p}=\Eb_1\left\|\left(\tr \left( \sum\limits_{i=1}^{m/2}e_ie_i^T\right)\right)^{1/2} \right\|_{4p}^{4p} \leq k^{2p} \Eb_1\|E\|_2^{4p}  
\end{align}
and by comparing \eqref{eq:II} with \eqref{eq:I}. Summing up the bounds on $(I),(II),(III)$, we get
$$M\leq 2(4p)^{2p}(I)+(4p)^{4p}(III) \leq c_3^pkmp^{6p}k^{2p}l^{2p} r_1^{8p}r_2^{4p}. $$
In the last inequality, we have used the fact that $r_1,r_2\geq 1$ which can be straightforwardly verified from their definitions.
Inserting the above bound of $M$ in \eqref{eq:decouple} we obtain
\[
\Eb \|F^TAF\|^{4p}_{4p}\leq c_4^{2p} kmp^{6p}k^{2p}l^{2p} r_1^{8p}r_2^{4p}.
\]
Using the Matrix Chebyshev inequality $P(\|H\|>t)\leq t^{-p} \mathbb{E} \|H\|_{p}^p$ (see e.g.,\cite{MT14}), it can be concluded that
\begin{equation}\label{eq:failure}
P(\|F^TAF\|\geq cl) \leq \frac{1}{(cl)^{4p}} c_4^p kmp^{6p}k^{2p}l^{2p} r_1^{8p}r_2^{4p}\leq m^2 \left(\frac{c_4}{c^2}  \frac{kp^3r_1^4r_2^2}{l}\right)^{2p} \leq e^{-(l/(c_5k))^{1/3}+2\log m},
\end{equation}
where $c_5(c,r_1,r_2)=c_4r_1^4r_2^2/(ec^2)$. For the last inequality, we have set $p=\left(\frac{c^2l}{c_4e r_1^4r_2^2 k }\right)^{1/3}$, and the condition  $p\geq 2$ in Lemma \ref{lm: MR} requires $l \geq 8c_5k $.
For the probability in \eqref{eq:failure} to be bounded by $\epsilon/2$ so that together with \eqref{eq:failure1} the total probability of failure adds up to $\epsilon$, we need $ l\geq c_5k \log^3 2m/\epsilon$. Putting all the requirement on $l$ together we get $m\geq l\geq \max\{2c_5 k\log^3 2m/\epsilon, 8c_5k\}$. 
\end{proof}

\begin{proof}[Proof of Theorem \ref{thm:main2}]
Based on whether $H$ contains the first column of the DFT matrix, i.e., the column of all ones, we divide the proof into two cases. \\
\emph{Case 1}: if $H$ does not contain $e = [1,...,1]^T$ as a column, then $e^T H =0$.  We apply Theorem \ref{thm:main1} to $H$ by setting $r_1=1$, $r_2=\sqrt 2$ and  $c=\eta/2$. The assumption $l\geq c_1 \frac{k}{\eta^2}  \log^3 (m/ \epsilon)$ (with a large enough $c_1$, say, $c_1\geq 36$) then ensures the assumptions of Theorem \ref{thm:main1} to be satisfied. Applying Theorem \ref{thm:main1} yields
\[
P\left(\sigma_{\min}(V_l^T F) \leq l\left(1-\eta \right)\right) \leq P\left(\sigma_{\min}(V_l^T F) \leq l\left(1-c-2\sqrt{\frac{k}{l}} \log\frac{4 k}{ \epsilon}\right)\right)\leq \epsilon,
\]
where the first inequality used the fact that $l\geq c_1 \frac{k}{\eta^2}  \log^3 (m/ \epsilon)$. \\
In a similar way, we can get
$$P\left(\sigma_{\max}(V_l^T F)\geq l\left(1+\eta \right)\right) \leq \epsilon.$$
\emph{Case 2}:   if $\Phi$ contains the all-ones column $e$, and without loss of generality suppose it is the first column of $\Phi$.  For simplicity of notation, set the frequently used quantity $V_l^TV_l$ as a single variable $G=V_lV_l^T$.  As in the proof of Theorem \ref{thm:main1}, we use $\widetilde{A}$ to denote the diagonal matrix that coincides with $V_lV_l^T$ on the main diagonal, so $A=G-\widetilde{A}$ would be the one containing the off-diagonal entries of $G$. In addition, let $v_i$ be the $i$th row of $V_l$, $a_i$ be the $i$th column of $A$ and $C$ be the matrix having $e^TAe \equiv \sum\limits_{i=1}^m \sum\limits_{j=1}^m A_{i,j}$ as the first element and 0 otherwise.  Let $E=[e,0]$ and use $\widetilde{F}$ to denote the matrix of $F$ with the first column set to zeros, i.e. $\widetilde{F}=F-E$. We decompose the target quantity $Y=F^TV_lV_l^TF$ as follows
\begin{equation}\label{eq:decomposition}
Y=F^T\widetilde{A}F+C+E^TGF-E^T\widetilde{A}F+F^TGE-F^T\widetilde{A}E+\widetilde{F}^T A\widetilde{F}.
\end{equation}
We shall bound the spectral norm of $Y$ by bounding each component of the above decomposition. Notice that the bounds of $\|F^T\widetilde{A}F\|$ and $\|\widetilde{F}^T A\widetilde{F}\|$ were found in Lemma \ref{lm:cond} and Theorem \ref{thm:main1}.  The bound of $\|E^T\widetilde{A}F\|$ can be directly estimated from the matrix Bernstein inequality using the facts  
\begin{itemize}
\item $\|E^T\widetilde{A}F\|\leq \|e^T\widetilde{A}F\|=\|\sum_{k=1}^m \|v_k\|^2 f_k\|$, 
\item  $\mathbb{E} M_k=0$, $\max_k \|M_k\| \leq \frac{\sqrt {kl}}{m}$, $\max\{ \|\sum\limits_k \mathbb{E}M_k^TM_k\|^2,\sum\limits_k \mathbb{E} M_kM_k^T\|^2 \} \leq kl$, where  $M_k := \|v_k\|^2 f_k$.
\end{itemize} 
These imply
\begin{equation}\label{eq:EAF} P\left( \|E^T\widetilde{A}F\| \geq t\right )\leq 2k \exp\left( -\frac{t^2}{3kl+ 2\sqrt {kl} t/m }\right),
\end{equation}
and 
$$P( \|E^T\widetilde{A}F\| \geq 3\sqrt {kl}\log \frac{2k}{\epsilon} )\leq \epsilon. $$
by setting $t=3\sqrt {kl}\log \frac{2k}{\epsilon}  $.

The following is devoted to estimating the quantity
$$ \| C+ E^TGF+F^TGE\| = \max\limits_{x \in \mathbb{C}^k}  \frac{x^T ( C+ E^TGF+F^TGE)x}{\|x\|^2}.$$ \\
For any fixed $x\in \mathbb{C}^k$, we denote its first element by $x_1$ and rest by $x_2$. Then 
\begin{equation}\label{eq:C}
x^TCx= \sum_k \langle e, g_k\rangle x_1^2-lx_1^2.
\end{equation}
By Proposition \ref{lm:MB},
\[P(\|E^TGF\|_2 \geq t ) \leq 2k \exp  \left( -\frac{t^2}{3k\sum_k |\langle e, g_k\rangle |^2+2\sqrt k \max_k |\langle e,g_k \rangle |}\right),
\]
where $g_k$ is the $k$th column of $G$.
Therefore, by letting $$t_0= \sqrt{6k\log \frac{4k}{\epsilon}} \sqrt{\sum_k | \langle e, g_k\rangle|^2}+ 4\log \frac{4k}{\epsilon} \sqrt k  \max_k |\langle e, g_k\rangle |, $$
the above inequality reduces to
\begin{equation}\label{eq:P} P\left(\left\| \sum_{k} \langle e, g_k\rangle f_k \right\| \geq t_0\right) \leq \frac{\epsilon}{2}.
\end{equation}
It is easy to verify that
\begin{equation}\label{eq:11}
|\langle e, g_k\rangle| \leq \sqrt m \| g_k\|_2 =  \sqrt m \| V_l v_k\|_2 \leq \sqrt{m} \sqrt{\frac{ 2l}{m}} = \sqrt {2l},   
\end{equation}
where $v_k$ is the $k$th row of $V_l$.
We also have the relations  
\begin{equation} \label{eq:2to1}
\sum_k | \langle e, g_k\rangle|^2 \leq \sum_k | \langle e, g_k\rangle |,
\end{equation}
and
\begin{equation}\label{eq:2l} \sum_k | \langle e, g_k\rangle | \geq 2l,
\end{equation}
which we will soon prove. 
Using these inequalities in the definition of $t_0$, we have $$ t_0 \leq \sqrt{6k \log \frac{4k}{\epsilon}\sum_k | \langle e, g_k\rangle|}+ 4 \log \frac{4k}{\epsilon} \sqrt {2kl}.$$

Inserting the above bound of $t_0$ into \eqref{eq:P} and using \eqref{eq:C}, we obtain with probability over $1-\epsilon$,

\begin{align}&x^T ( C+ E^TGF+F^TGE)x  \notag \\
&\geq (-l+\sum_k \langle e, g_k\rangle )|x_1|^2-2|x_1|\left( \sqrt{6k \log \frac{4k}{\epsilon}\sum_k | \langle e, g_k\rangle|}+ 4 \log \frac{4k}{\epsilon} \sqrt {2kl}\right) \|x_2\| \notag \\
&\geq (\frac{1}{2}\sum_k \langle e, g_k\rangle )|x_1|^2-2\sqrt {6 k} |x_1|\|x_2\| \sqrt{\log \frac{4k}{\epsilon} \sum_k | \langle e, g_k\rangle|}- 8|x_1|\|x_2\| \log \frac{4k}{\epsilon} \sqrt {2kl} \notag \\
&\geq \frac{1}{2}  \left(\sqrt{\sum_k \langle e, g_k\rangle }|x_1|- 2\sqrt {6k}\sqrt{\log \frac{4k}{\epsilon}}\|x_2\|_2 \right)^2-12k\log \frac{4k}{\epsilon} \|x_2\|^2-8|x_1|\|x_2\|  \sqrt {2kl} \log \frac{4k}{\epsilon} \notag \\
&\geq -c_2\|x\|^2 \log \frac{4k}{\epsilon} \sqrt {2kl}, \label{eq: 2}
\end{align}
with some $c_2>0$. 
Inserting \eqref{eq: 2}, \eqref{eq:2to1} and \eqref{eq:EAF} to \eqref{eq:decomposition}, we obtain the result of the current theorem
$$ \sigma_{\min} (F^TV_l^TV_lF) \geq l \left(1-c-c_3\sqrt {\frac{k}{l} }\log\frac{4k}{\epsilon} \right) \geq  l(1-\eta), $$
with probability over $1-2\epsilon$, provided that $l\geq \frac{c_4 k }{\eta^2 } \log^3 \frac{k}{\epsilon}$ for some absolute $c_4 \geq 3c_3$ that ensures the assumptions of Theorem \ref{thm:main1} to be satisfied.

Now we go back to prove  \eqref{eq:2to1} and \eqref{eq:2l}. \eqref{eq:2to1} is due to the following calculations
\begin{align}
& \sum_k | \langle e, g_k\rangle|^2 =\sum_k \left|\left(\sum_l v_l\right)^Tv_k \right|^2 =\sum_k  v_k ^T \left(\sum_l v_l\right) \left(\sum_l v_l\right)^Tv_k  =\sum_k \tr \left( v_k ^T \left(\sum_l v_l\right) \left(\sum_l v_l\right)^Tv_k\right) \notag\\ &=\sum_k \tr\left(v_k v_k ^T \left(\sum_l v_l\right) \left(\sum_l v_l\right)^T\right) =\tr\left(\left(\sum_k v_k v_k ^T\right) \left(\sum_l v_l\right) \left(\sum_l v_l\right)^T\right)  \notag\\ 
&=\tr\left(V_l^TV_l \left(\sum_l v_l\right) \left(\sum_l v_l\right)^T\right)=\tr\left(\left(\sum_l v_l\right) \left(\sum_l v_l\right)^T\right) =\tr\left(\left(\sum_l v_l\right)^T \sum_l v_l \right) \notag\\
&=\left(\sum_l v_l\right)^T \left(\sum_l v_l\right) \leq \sum_k | \langle e, g_k\rangle|. \notag
\end{align}

\normalfont

To verify \eqref{eq:2to1}, let $w_1$ be the first column of $V_l$ whose explicit form is given in Proposition \ref{pro:eigD}. We have
\begin{align*}
e^Tw_1 &=  \sqrt{\frac{2}{m+1/2}} \sum_{k=1}^m  \cos\left(\frac{(m-k+1/2)\pi}{2m+1}\right) \\
&=  \sqrt{\frac{2}{m+1/2}} \sum_{k=1}^m  \sin\left(\frac{k\pi}{2m+1}\right) \\
&= \sqrt{\frac{2}{m+1/2}} \frac{\cos \left(\frac{\pi}{2(2m+1)}\right)}{2\sin\left(\frac{\pi}{2(2m+1)}\right)} \\
&> \sqrt{\frac{2}{m+1/2}} \frac{1-\frac{1}{2}\left(\frac{\pi}{2(2m+1)}\right)^2}{\frac{\pi}{2m+1}}  \\
&\geq \frac{\sqrt{2m+1}}{\pi}.
\end{align*}
In the second to last inequality, we have used the basic facts: $\sin(\theta) < \theta$, and $\cos(\theta) >1- \theta^2/2$ for $\theta>0$.
Hence,
$$ \sum_k | \langle e, g_k\rangle | \geq  (\sum_l v_l)^T (\sum_l v_l)) \geq |\langle e,w_1 \rangle |^2 \geq \frac{2m+1}{\pi^2} \geq 2l, $$
where we have used the assumption $m\geq \pi^2 l$.

\end{proof}
\begin{proof}[Proof of Theorem \ref{thm:freq}]
If $\delta_{1} (\omega_2)=0$, then the result directly follows from Theorem \ref{thm:main1}. If $\delta_{1} (\omega_1)=0$ and $\delta_{1} (\omega_2)\neq 0$, then $(H^{\omega_1})^T H^{\omega_2} =  (H^{\omega_1})^T [0,H^{\omega_2\backslash \{1\}}]$. So
\[
0 = \sigma_{\min}((H^{\omega_1})^T H^{\omega_2}) \leq  \| (H^{\omega_1})^T H^{\omega_2} \|_2 \leq \| (H^{\omega_1})^T H^{\omega_2\backslash \{1\}}\|_2 \leq C \sqrt {l/N},
\]
where the last inequality used the result of Theorem \ref{thm:main1}. 
If $\delta_{1} (\omega_1) \delta_{1} (\omega_2) = 1$, then 
$$(H^{\omega_1})^T H^{\omega_2} = \left[ \begin{matrix} 1 & 0  \\ 0 & H^{\omega_1 \backslash \{1\}} )^T H^{\omega_2 \backslash \{1\} } \end{matrix} \right].$$
Hence we have,
\[
1\geq \sigma_{\max}( (H^{\omega_1})^T H^{\omega_2}) \geq  \sigma_{\min}( (H^{\omega_1})^T H^{\omega_2}) \geq  \sigma_{\min}( (H^{\omega_1\backslash \{1 \} })^T H^{\omega_2 \backslash \{1 \}}) \geq c \sqrt {l/N}.
\]
The result then follows. 
\end{proof}
\subsection{Proof of Theorem \ref{thm:cs} } 
\begin{proof} [Proof of Theorem \ref{thm:cs} ]
The proof essentially uses the technique  in \cite{CRT}. We need to modify the tube constraint to deal with the slightly weaker RIP we have in this case.  \\
Tube constraint: from the feasibility of $x$ and $\hat{x}$ and the triangle inequality, we have
$$\|D^{-1}\widetilde{A} (x-\hat{x}) \|_2\leq \sqrt m \delta.$$ 
Let $USV^T$ be the SVD of $D^{-1}$, then 
\[ \sqrt m \delta \geq \|S_l V_l^T \widetilde{A} (x-\hat{x}) \|_2\geq \sigma_l(D^{-1})\sqrt l \|\frac{1}{\sqrt l} V^T_l  \widetilde{A} (x-\hat{x}) \|_2\geq C \frac{m}{\sqrt l}\|\frac{1}{\sqrt l} V^T_l  \widetilde{A} (x-\hat{x}) \|_2,  \]
where  we again used the fact that 
\begin{equation}\label{eq:eigD}
c_1\left(\frac{m}{l}\right)^r \leq \sigma_l(D^{-r}) \leq c_2\left(\frac{m}{l}\right)^r,
\end{equation} derived in \cite{GLPSY13}.
Let $\hat{A}= \frac{1}{\sqrt l}  V^T_l  \widetilde{A}  $, and $h=\hat{x}-x$, then the above inequality reduces to
\begin{equation}\label{eq:1}
\|\hat A h\|_2 \leq C \delta \sqrt{\frac{l}{m}}.
\end{equation}
Now if $\hat A$ was RIP, then we can invoke Proposition \ref{pro:2.2} to finish the proof. However, since the matrix $\widetilde{A}$ contains the columns of all ones, any sub-matrix of $\hat{A}$ that contains that column cannot be shown to have RIP with a satisfactory constant. Without loss of generality, suppose the all-one column is the first column of $\widetilde{A}$. Fix a $\delta >0$, Theorem \ref{thm:main2} and a union bound on the probability imply that there exists a $c_{\delta}$ such that provided $l\geq c_\delta k^4 \log^3\frac{N}{\epsilon}$, it holds for any $K \in \{1,...,N\} $ with $|K|\leq 4k+1$ and $\{1\}\notin K$, that
\[
(1-\delta)\|x\|^2 < \|\hat A_Kx\|^2_2 \leq (1+\delta)\|x\|^2,
\]
and for any $K$ with $K \subseteq \{1,...,N\}$ and $|K|\leq 4k+1$, that
\[
(1-\delta) \|x\|_2^2 <\|\hat A_Kx\|_2^2 ,
\]
with probability over $1-\epsilon$.
Next we show that the loss of uniform RIP upper bound does not affect the proof of Theorem 1 of \cite{CRT}. 
Let $K_0$ be the support of the original signal $x$, and let $T_0= K_0 \cup \{1\}$. This means that, we always assume the first index to be in the support of $x$, although the magnitude may be zero. Now since $\|\hat{x}\|_1 \leq \|x\|_1$, the cone constraint gives
\[
\| h_{K_0^c} \|_1 \leq \|h_{K_0} \| _1 ,
\] 
which in turn implies 
\[
\| h_{T_0^c} \|_1 \leq \|h_{T_0} \| _1 .
\]
Let $M =3|T_0|$ and divide $T_0^c$ rearranged in descending order of magnitude of $h_{T_0^c}$  into subsets of size $M$,  $T_j=\{n_l, (j-1)M+1\leq l\leq jM \} $, and let $T_{01}=T_0 \cup T_1$. It is easy to check that the cone constraint gives 
\begin{equation}\label{eq:2}
\|h\|_{2}^2 \leq \frac{4}{3} \|h_{T_{01} }\|_2^2,
\end{equation}
and that the tube constraint implies 
\begin{align*}
\|\hat A h \|_2 & = \|\hat A_{T_{01}} h_{T_{01}} + \sum\limits_{j\geq 2} \hat  A_{T_j}h_{T_j} \|_2 \\& \geq 
 \|\hat  A_{T_{01}} h_{T_{01}}\|_2 - \sum\limits_{j\geq 2} \|\hat  A_{T_j}h_{T_j} \|_2 \\
& \geq \sqrt {1-\delta} \|h_{T_{01}}\|_2 - \sqrt {1+\delta} \sum\limits_{j\geq 2} \|h_{T_j}\|_2.
\end{align*}
Note that we only used the lower RIP bound of $A_{T_{01}}$, where $T_{01}$ is the partition of $\hat{A}$ that contains the first column. This inequality with the following inequality proved in \cite{CRT} 
\[
\sum\limits_{j\geq 2} \|h_{T_j}\|_2\leq \sqrt{1/3} \|h_{T_{0}}\|_2,
\]
give \[
\|h_{T_0} \|_2 \leq C_\delta \|Ah\|_2.
\]
Combing this with \eqref{eq:1} and \eqref{eq:2} proofs the statement of the theorem.
\end{proof}
\begin{remark}\label{rm:1}
Theorem \ref{thm:cs} and Theorem \ref{thm:main2} are both established based on the matrix Rosenthal inequality (Proposition \ref{lm: MR}). One may also uses the vector Rosenthal inequality with a similar argument and get an improvement on the order of $k$ in \eqref{eq:csbound} of Theorem \ref{thm:cs} from 2 to 3/2.
\end{remark}
\section{Numerical experiments}
The purpose of this section is to observe the asymptotic behavior of the distortion proved in Theorem \ref{cor:optimal} and Theorem \ref{thm:cs} in simulations.

The first experiment is designed to demonstrate the necessity of performing random permutations on the Fourier measurements. Suppose we have a harmonic frame $H^{\omega}$ that consists of the last 10 columns of the $512 \times 512$ DFT matrix. Using quantization step size  $\delta=0.1$, and letting the number of measurement $m$ range from 100 to 500, we test on two scenarios: direct quantization (quantizing $H^{\omega}x$) versus randomized quantization (quantizing $H_{\sigma}^{\omega}x$), where $H_{\sigma}^{\omega}$ is the same as in Theorem \ref{thm:freq} being the random selection of rows of $H_{\sigma}^{\omega}$ according to $\sigma$. For both settings, the Sobolev dual decoder is employed for reconstruction. Each time we draw a random signal $x$ from the unit ball in $\mathbb{C}^{10}$ uniformly, obtain its quantization through the first order $\Sigma\Delta$ quantizer, reconstruct $x$ from $H_{sob}^{\omega}$ and $H_{\sigma,sob}^{\omega}$, respectively, and compute the distortion. Figure \ref{fig:1} displays the worse-case reconstruction error for both cases as a function of $m$ over 200 independent draws of random signals. Clearly, the reconstruction error of the direct quantization shows no tendency to decay  as more samples are taken, while that of the randomized quantization decays polynomially.
\begin{figure}[htbp]
\begin{center}
\includegraphics[scale=0.3]{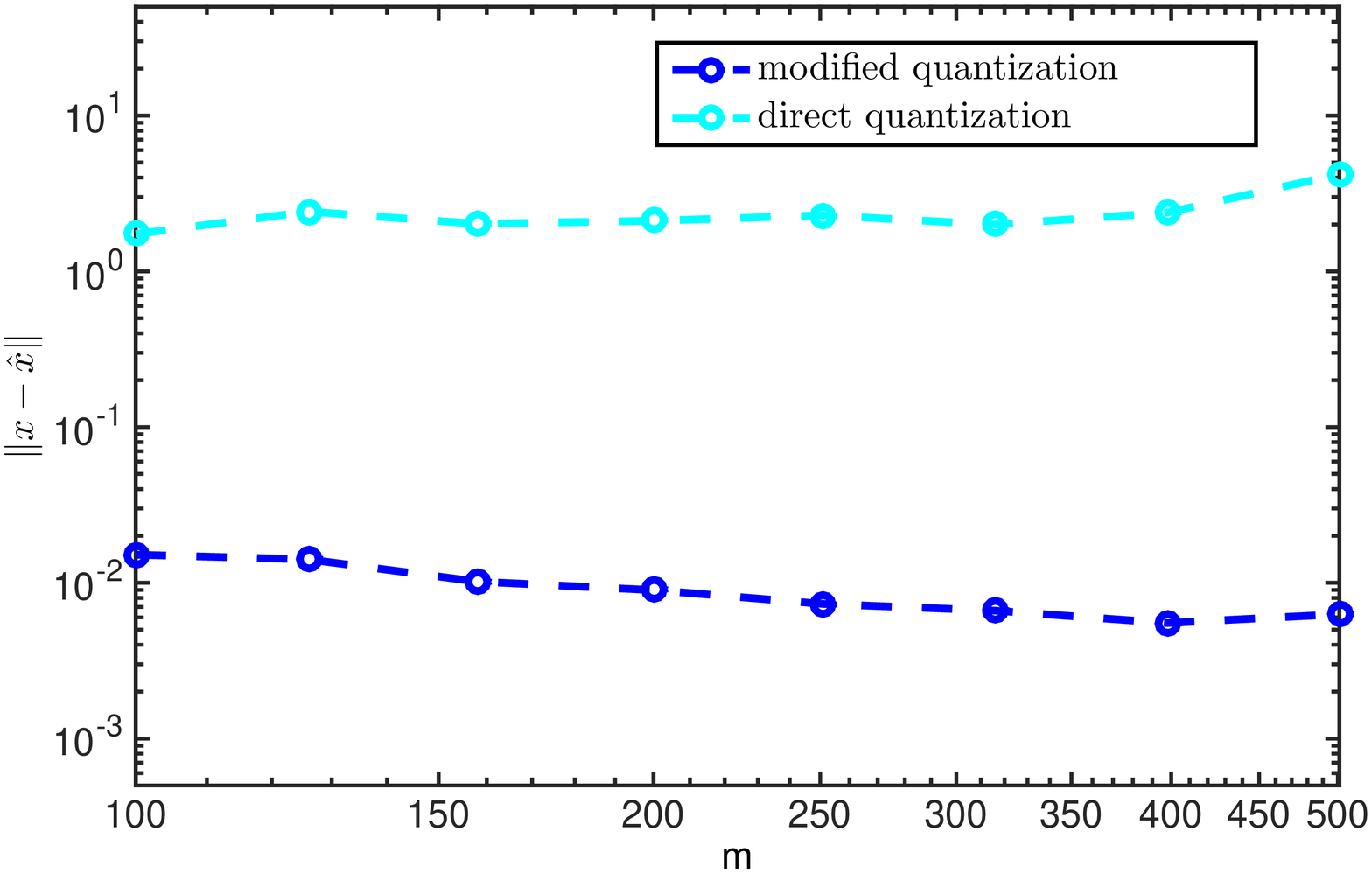}
\caption{Reconstruction error of the direct quantizer versus the randomized quantizer. }
\label{fig:1}
\end{center}
\end{figure}

Generated from the same parameter setting, Figure \ref{fig:2} confirms the 1/2 order convergence rate for the distortion with respect to $m$ derived in Theorem \ref{cor:optimal}. In addition, the displayed $3/2$ order convergence rate in this figure for the second order $\Sigma\Delta$ quantization provides evidences for the validity of Conjecture \ref{conj}. Each point on this figure is based on taking the maximum error over 400 independent draws of $x$, and the figure is plotted on a log-log scale.
\begin{figure}[htbp]
\begin{center}
\includegraphics[scale=0.3]{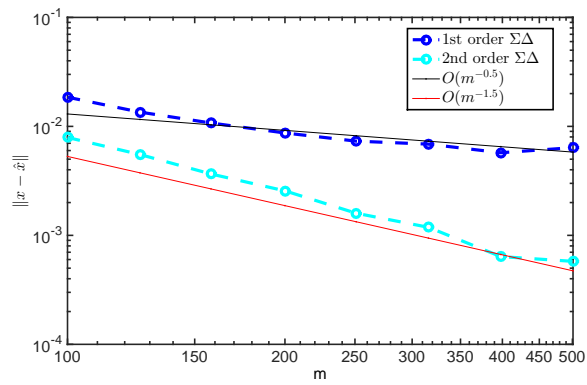}
\caption{Worst case reconstruction error for the first and second order $\Sigma\Delta$ quantization with harmonic frames. The error is plotted as a function of $m$. }
\label{fig:2}
\end{center}
\end{figure}

\begin{figure}[htbp]
\begin{center}
\includegraphics[scale=0.3]{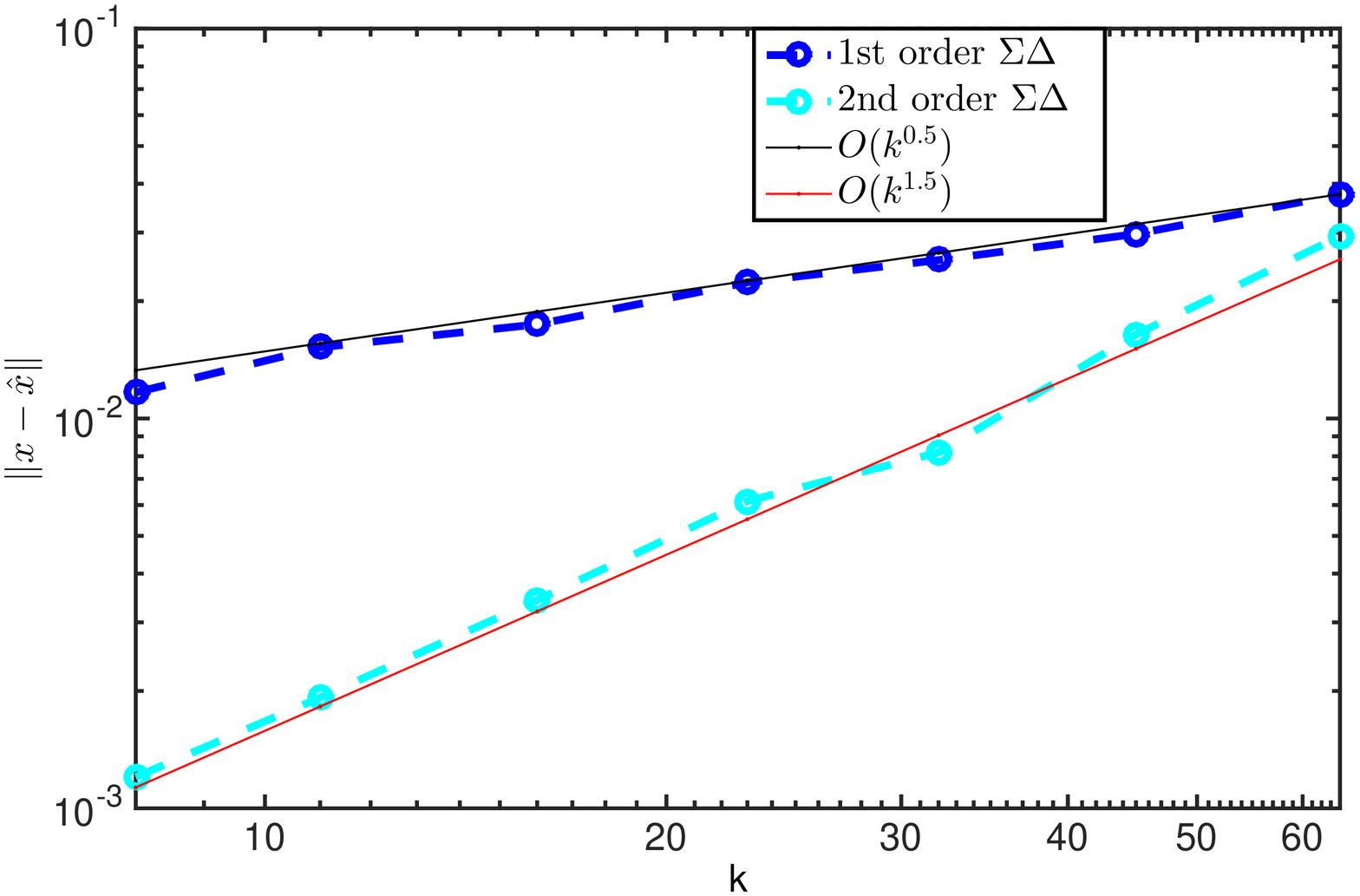}
\caption{Worst case reconstruction error for the first and second order $\Sigma\Delta$ quantization with harmonic frames. The error is plotted as a function of $k$.}
\label{fig:3}
\end{center}
\end{figure}
We then test the optimality of the upper bound in Theorem \ref{cor:optimal} with respect to $k$.  To this end, we fix $m$ to be 512, and let $k$ grow from 8 to 64. Figure \ref{fig:3} displays the worse case error as a function of $k$ over 20 draws of $x$ for a fixed $\sigma$. The plot suggests that our result is optimal. 
\begin{figure}[h]
\begin{center}
\includegraphics[scale=0.3]{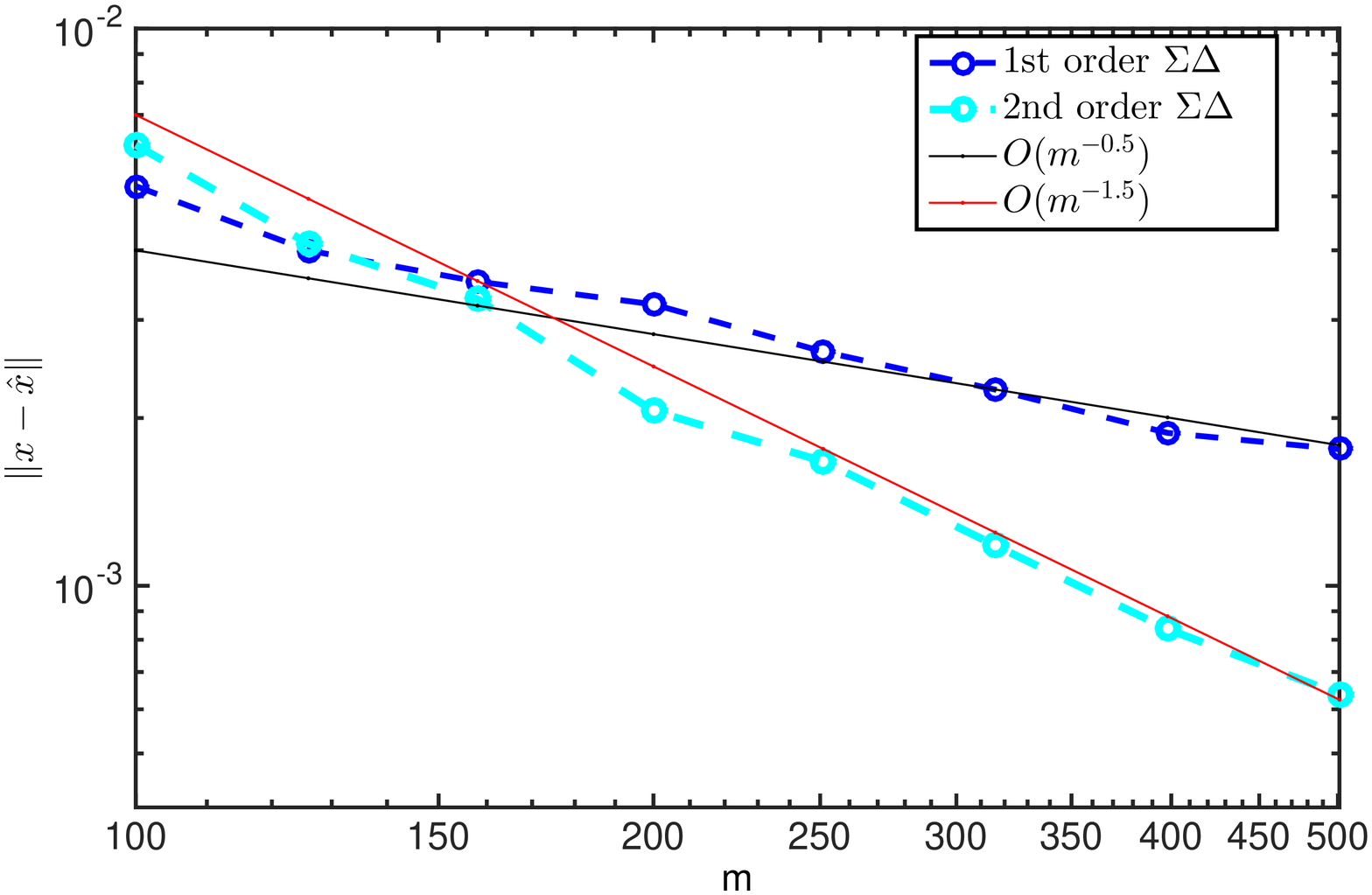}
\caption{Worst case reconstruction error for the first and second $\Sigma\Delta$ quantization with partial Fourier ensembles. The error is plotted as a function of $m$.}
\label{fig:4}
\end{center}
\end{figure}

Next we do similar tests for results obtained in the compressed sensing setting (Theorem \ref{thm:cs}).  We found that the asymptotic order is optimal for $m$.  Specifically, set the sparsity level $k=10$ and quantization step size $\delta=0.1$. Let the length of signals be $N=512$ and the number of measurements $m$ range from 100 to 1000.  For each draw of random permutations, we generate 20 signals uniformly from the unit $\ell_2$ ball, apply the $\Sigma\Delta$ quantization and conduct the reconstruction in \eqref{eq:l1}. Each point in Figure \ref{fig:4} represents the worst case error over 20 draws of random permutations and 20 draws of random signals for each permutation. The error shows a decay rate of exactly $1/2$ with respect to $m$ for the first order quantization as proved in Theorem \ref{thm:main1} and a rate of $3/2$ for the second order quantization as suggested in Conjecture \ref{conj}.
\section{Acknowledgement}
The author would like to thank Rayan Saab, {\"O}zgur Y{\i}lmaz and Joe-Mei Feng for valuable discussion about the topic. R. ~Wang was funded in part by an NSERC Collaborative Research and Development Grant DNOISE II (22R07504). 
\bibliographystyle{plain}
\bibliography{DCCbib}

\begin{thebibliography}{10}

\bibitem{ALPV14}
A.~Ai, A.~Lapanowski, Y.~Plan, and R.~Vershynin.
\newblock One-bit compressed sensing with non-gaussian measurements.
\newblock {\em Linear Algebra and its Applications}, 441:222--239, 2014.

\bibitem{baraniuk2014exponential}
R.~Baraniuk, S.~Foucart, D.~Needell, Y.~Plan, and M.~Wootters.
\newblock Exponential decay of reconstruction error from binary measurements of
  sparse signals.
\newblock {\em arXiv preprint arXiv:1407.8246}, 2014.

\bibitem{benedetto2006sigma}
J.J. Benedetto, A.M. Powell, and {\"O}.~Y{\i}lmaz.
\newblock {Sigma-delta ($\Sigma\Delta$) quantization and finite frames}.
\newblock {\em IEEE Trans.~Inf.~Theory}, 52(5):1990--2005, 2006.

\bibitem{blum:sdf}
J.~Blum, M.~Lammers, A.M. Powell, and {\"O}.~Y{\i}lmaz.
\newblock {Sobolev duals in frame theory and sigma-delta quantization}.
\newblock {\em J.~Fourier Anal.~Appl.}, 16(3):365--381, 2010.

\bibitem{BPA2007}
B.G. Bodmann, V.I. Paulsen, and S.A. Abdulbaki.
\newblock {Smooth frame-path termination for higher order sigma-delta
  quantization}.
\newblock {\em J.~Fourier Anal.~Appl.}, 13(3):285--307, 2007.

\bibitem{CZ14}
T.~Cai and A.~Zhang.
\newblock Sparse representation of a polytope and recovery of sparse signals
  and low-rank matrices.
\newblock {\em IEEE transactions on information theory}, (122-132), 2014.

\bibitem{CRT}
E.~J. Cand\`es, J.~Romberg, and T.~Tao.
\newblock Robust uncertainty principles: Exact signal reconstruction from
  highly incomplete frequency information.
\newblock {\em IEEE Trans. Inf. Theory}, 52(2):489--509, 2006.

\bibitem{C08}
J.~E. Candes.
\newblock The restricted isometry property and its implications for compressed
  sensing.
\newblock {\em Comptes Rendus Mathematique}, 346(9):589--592, 2008.

\bibitem{CG15}
E.~Chou and G{\"u}nt{\"u}rk.
\newblock Distributed noise-shaping quantization: I. beta duals of finite
  frames and near-optimal quantization of random measurements.
\newblock {\em arXiv preprint arXiv:1405.4628}, 2015.

\bibitem{daub-dev}
I.~Daubechies and R.~DeVore.
\newblock {Approximating a bandlimited function using very coarsely quantized
  data: a family of stable sigma-delta modulators of arbitrary order}.
\newblock {\em Ann. Math.}, 158(2):679--710, 2003.

\bibitem{DGK10}
P.~{D}eift, C.~S. G{\"u}nt{\"u}rk, and F.~{K}rahmer.
\newblock An optimal family of exponentially accurate one-bit sigma-delta
  quantization schemes.
\newblock {\em Comm.~Pure Appl. Math.}, 64(7):883--919, 2011.

\bibitem{Feng2014}
J.~Feng and F.~Krahmer.
\newblock An {RIP}-based approach to {$\Sigma \Delta$} quantization for
  compressed sensing.
\newblock {\em IEEE Signal Process.~Lett.}, 21(11):1351--1355, 2014.

\bibitem{Foucart13}
S.~Foucart.
\newblock Stability and robustness of $\ell_1$-minimizations with {Weibull}
  matrices and redundant dictionaries.
\newblock {\em Linear Algebra and its Applications}, 441:4--21, 2014.

\bibitem{G-exp}
C.S. G{\"u}nt{\"u}rk.
\newblock One-bit sigma-delta quantization with exponential accuracy.
\newblock {\em Comm.~Pure Appl. Math.}, 56(11):1608--1630, 2003.

\bibitem{GLPSY13}
C.S. G{\"u}nt{\"u}rk, M.~Lammers, A.M. Powell, R.~Saab, and {\"O}.~Y{\i}lmaz.
\newblock Sobolev duals for random frames and sigma-delta quantization of
  compressed sensing measurements.
\newblock {\em Foundations of Computational mathematics}, 13(1):1--36, 2013.

\bibitem{inose1963unity}
H.~Inose and Y.~Yasuda.
\newblock {A unity bit coding method by negative feedback}.
\newblock {\em Proceedings of the IEEE}, 51(11):1524--1535, 1963.

\bibitem{iwen}
M.~Iwen and R.~Saab.
\newblock Near-optimal encoding for sigma-delta quantization of finite frame
  expansions.
\newblock {\em J. Fourier Anal. Appl.}, 19(6):1255--1273, 2013.

\bibitem{jacques2013robust}
L.~Jacques, J.~N. Laska, P.~T. Boufounos, and R.~G. Baraniuk.
\newblock Robust 1-bit compressive sensing via binary stable embeddings of
  sparse vectors.
\newblock {\em IEEE Trans.~Inf.~Theory}, 59(4):2082--2102, 2013.

\bibitem{KSW14}
K.~Knudson, R.~Saab, and R.~Ward.
\newblock One-bit compressive sensing with norm estimation.
\newblock {\em arXiv preprint arXiv:1404.6853}, 2014.

\bibitem{KSW12}
F.~Krahmer, R.~Saab, and R.~Ward.
\newblock Root-exponential accuracy for coarse quantization of finite frame
  expansions.
\newblock {\em IEEE Trans.~Inf.~Theory}, 58(2):1069 --1079, February 2012.

\bibitem{KSY13}
F.~Krahmer, R.~Saab, and {\"O}~Y{\i}lmaz.
\newblock {Sigma-Delta} quantization of sub-{G}aussian frame expansions and its
  application to compressed sensing.
\newblock {\em Information and Inference}, 3(1):40--58, 2013.

\bibitem{lammers:adf}
M.~Lammers, A.M. Powell, and {\"O}.~Y{\i}lmaz.
\newblock {Alternative dual frames for digital-to-analog conversion in
  sigma--delta quantization}.
\newblock {\em Adv. Comput. Math.}, pages 1--30, 2008.

\bibitem{MT14}
L.~Mackey, M.~Jordan, R.~Chen, B.~Farrell, and J.~Tropp.
\newblock Matrix concentration inequalities via the method of exchangeable
  pairs.
\newblock {\em The Annals of Probability}, 42(3):906--945, 2014.

\bibitem{Maricic}
D.~Maricic.
\newblock Image sensors employing oversampling sigma-delta analog-to-digital
  conversion with high dynamic range and low power.
\newblock {\em Diss. University of Rochester}, 2011.

\bibitem{PV13}
Y.~Plan and R.~Vershynin.
\newblock One-bit compressed sensing by linear programming.
\newblock {\em Comm.~Pure Appl. Math.}, 66(8):1275--1297, 2013.

\bibitem{PV13b}
Y.~Plan and R.~Vershynin.
\newblock Robust 1-bit compressed sensing and sparse logistic regression: A
  convex programming approach.
\newblock {\em IEEE Transactions on Information Theory}, 59(1):482--494, 2013.

\bibitem{PSY13}
A.~M. Powell, R.~Saab, and {\"O}.~Y{\i}lmaz.
\newblock Quantization and finite frames.
\newblock In P.~G. Casazza and G.~Kutyniok, editors, {\em Finite Frames}, ANHA,
  pages 267--302. Birkh{\"a}user Boston, 2013.

\bibitem{RV08}
M.~Rudelson and R~Vershynin.
\newblock On sparse reconstruction from fourier and gaussian measurements.
\newblock {\em Communications on Pure and Applied Mathematics},
  61(8):1025--1045, 2008.

\bibitem{SWY}
R.~Saab, R.~Wang, and {\"O}.~Y{\i}lmaz.
\newblock Quantization of compressive samples with stable and robust recovery.
\newblock {\em preprint arXiv:1504.00087}, 2015.

\bibitem{tropp}
J.~Tropp.
\newblock On the conditioning of random subdictionaries.
\newblock {\em Applied and Computational Harmonic Analysis}, 25(1):1--24, 2008.

\bibitem{Ver}
R.~Vershynin.
\newblock A simple decoupling inequality in probability theory.
\newblock {\em preprint}, 2011.

\end{thebibliography}
\begin{appendix}

\end{appendix}
\end{document}